\documentclass[11pt]{article}

\usepackage[margin=1in]{geometry}
\usepackage{amsmath,amssymb,amsthm}
\usepackage{algorithm}
\usepackage[noend]{algpseudocode}
\usepackage{booktabs}
\usepackage{hyperref}
\usepackage{xcolor}
\usepackage{enumitem}

\algnewcommand\algorithmicupon{\textbf{upon}}
\algdef{SE}[UPON]{Upon}{EndUpon}[1]{\algorithmicupon\ #1 \algorithmicdo}{\algorithmicend\ \algorithmicupon}

\newtheorem{theorem}{Theorem}
\newtheorem{lemma}[theorem]{Lemma}
\newtheorem{proposition}[theorem]{Proposition}

\theoremstyle{definition}
\newtheorem{definition}{Definition}
\theoremstyle{remark}
\newtheorem{remark}{Remark}

\newcommand{\Select}{\mathsf{Select}}
\newcommand{\Valid}{\mathsf{Valid}}
\newcommand{\Auth}{\mathsf{Auth}}
\newcommand{\ECHO}{\textsc{Echo}}
\newcommand{\READY}{\textsc{Ready}}
\newcommand{\FREEZE}{\textsc{Freeze}}
\newcommand{\VOTE}{\textsc{Vote}}
\newcommand{\PC}{\mathsf{PC}}
\newcommand{\RC}{\mathsf{RC}}
\newcommand{\Cert}{\mathsf{Cert}}
\newcommand{\nil}{\bot}
\newcommand{\nslot}{\mathsf{ns}}
\newcommand{\rslot}{\mathsf{rs}}

\title{Resource Locks and Fast Unlock for Shared Balances\\
in Certificate-Based Payment Systems}
\author{Lingkang Shangguan, \textit{Student Member, IEEE}\\
  The University of Sydney\\
  \texttt{lsha0379@uni.sydney.edu.au}}
\date{Version 1.0}

\begin{document}
\maketitle

\begin{abstract}
Certificate-based payment systems such as FastPay settle a payment with one round of
signatures from a quorum of validators and avoid global ordering. The client collects the
signatures, and validators never talk to each other. We observe that this client-driven design
is complete only as long as no lock is ever undone. Each account has a single writer, a
conflict can only come from that writer's own misbehavior, and the account stays locked as
the penalty. Stablecoin payments routinely need balances that several accounts debit: treasuries,
merchant pools, allowances and subscriptions. There, honest accounts race, locks must be
undone, and a quorum certificate held privately by a client can be silently overridden. We
prove that no unlock rule is safe in this setting unless certificate formation is visible
to validators.

We present a hybrid protocol. Personal accounts keep the client-driven FastPay path
unchanged. A shared resource is a balance with its own version number and a set of
authorized accounts. Transactions on shared resources are submitted to any validator and
certified by a two-phase broadcast, $\ECHO$ then $\READY$. The admission check requires both the
sender's nonce and the resource's nonce to be free. Without contention, this path needs no
consensus. When concurrent debits split the votes in any proportion, validators detect the
split themselves, freeze the affected resource version, and resolve it with one agreement
instance scoped to that version. Losing transactions return their sender nonce. We prove
consistency, no overdraft, validity, finality and liveness. The protocol composes with the
PDCC reconfiguration protocol.
\end{abstract}

\section{Introduction}

Payment systems built on Byzantine consistent broadcast~\cite{cachin2011,fastpay,astro}
achieve sub-second settlement and horizontal scalability because validators never agree on a
global order. In FastPay~\cite{fastpay}, a client sends its signed transfer to all validators,
collects $2f+1$ signatures into a certificate, and broadcasts the certificate. Validators do
not communicate with each other, so each payment costs $O(n)$ messages.

The design rests on two facts. First, each account has a single writer, and the account nonce
orders that writer's transfers. Asset transfer with single-owner accounts has consensus number
one~\cite{guerraoui2019}. Second, an honest validator signs at most one transfer per nonce and never
revokes that signature. Together these make a certificate final the moment it exists, even if
only the client holds it, because no conflicting certificate can ever be formed. If the owner
equivocates, the account stays locked, and only the owner is harmed. Delegating
signature collection to the client is therefore safe, and it is what makes the system fast.

\paragraph{When client-driven certification breaks.}
Stablecoin workloads need balances debited by several accounts: a corporate treasury with
several payers, a merchant pool debited by refunds and payouts, delegated allowances and
subscriptions, and custodial wallets operated by many signing workers. Guerraoui et
al.~\cite{guerraoui2019} show that an account shared by $k$ owners has consensus number $k$.
Concretely, two honest accounts can debit the same balance at the same time and split the
validators' votes, so that neither reaches a quorum and the balance is stuck. This is not a
penalty anyone deserves. The lock must be undone.

Undoing a lock is exactly what the client-driven design cannot support. Before releasing a
stuck balance, validators must know that no transaction on it has already been certified. In
a client-driven design a certificate may exist only in a client's hands. A Byzantine validator
can sign one transaction privately and another publicly, so the validators' joint view can
show a split while a certificate exists. We prove (Proposition~\ref{prop:impossible}) that
no unlock rule over validators' reports is safe in this setting. This holds even if
validators broadcast their votes to each other. The protocol, not the client, has to record
that a certificate was formed.

\paragraph{Our design.}
We keep the fast path where it is safe and add visibility only where locks may be undone.
\begin{itemize}[noitemsep]
  \item \emph{Personal path.} A transaction that debits the sender's own account runs the
        FastPay protocol unchanged: client-driven, $O(n)$ messages, no unlock.
  \item \emph{Shared path.} A \emph{shared resource} has a balance, its own version number (the
        resource nonce) and a set of authorized accounts. A transaction debiting it can be
        submitted to any validator. A validator admits it only if the sender's nonce and the
        resource nonce are both free, and then broadcasts an $\ECHO$. A validator that has
        seen $2f+1$ echoes for a transaction stores them and broadcasts $\READY$, and
        $2f+1$ readies make the transaction final. This is Bracha's two-phase
        broadcast~\cite{bracha1987} applied per resource version, with no consensus.
  \item \emph{Fast unlock.} If a resource version is not resolved in time, validators freeze
        it, report what they echoed and stored, and run one validated agreement
        instance~\cite{cachin2001,vaba} for that version. The selection rule admits only a
        transaction whose echo quorum someone stored, which includes every final transaction.
        Losing transactions are dead, and their sender nonces are released for reuse.
\end{itemize}

\paragraph{Contributions.}
\begin{enumerate}[noitemsep]
  \item A characterization of when client-driven certification is complete (single writer,
        no unlock), and a proof that no unlock rule is safe without visible certification
        (\S\ref{sec:completeness}).
  \item A hybrid protocol with dual-nonce admission: personal transactions on the FastPay
        path, shared transactions on an $\ECHO$/$\READY$ path (\S\ref{sec:protocol}).
  \item Fast unlock: in-protocol detection, freeze and per-version agreement, with nonce
        release for losing transactions (\S\ref{sec:unlock}).
  \item Proofs of consistency, no overdraft, validity, finality and liveness
        (\S\ref{sec:proofs}).
  \item Integration with PDCC reconfiguration~\cite{pdcc} (\S\ref{sec:reconfig}).
  \item An open-source TLA$^+$ model, exhaustively checked with TLC, which also reproduces the
        counterexample of the impossibility result (\S\ref{sec:eval}).
\end{enumerate}

\section{Background}
\label{sec:background}

\paragraph{Certificate-based payments.}
FastPay~\cite{fastpay} and its successors settle each payment by Byzantine consistent
broadcast on a per-account basis; credits commute and need no ordering. PDCC~\cite{pdcc}
describes the resulting architecture as per-account certificate payments combined with a
periodic settlement consensus that commits checkpoints and reconfigurations.

\paragraph{Stuck accounts.}
An equivocating sender can split votes on its own nonce. FastPay leaves the account
stuck. Sui~\cite{sui} keeps such owned objects locked until the end of the epoch.

\paragraph{Reliable broadcast.}
Bracha's broadcast~\cite{bracha1987} includes an additional $\READY$ all-to-all phase after the echo phase to ensure that if an honest party delivers the message, then a quorum has seen a quorum of echo messages. Astro~\cite{astro} applies Byzantine reliable broadcast to build payment systems.

\paragraph{Validated agreement.}
In validated Byzantine agreement~\cite{cachin2001} all honest validators decide on a common value satisfying a predicate $\Valid$; PBFT-style protocols~\cite{pbft,hotstuff} implement this in the partial synchrony model with a transferable decision proof and randomized protocols~\cite{vaba} in the asynchronous model.

\section{Model}
\label{sec:model}

\paragraph{Validators, faults, network.}
In a configuration, a committee of $n = 3f+1$ validators runs the protocol; at most $f$ are
Byzantine. A \emph{quorum} is any $2f+1$ validators; two quorums share at least one honest
validator, and a quorum contains at least $f+1$ honest validators. Accounts and clients may be
Byzantine and collude with Byzantine validators. Signatures are unforgeable. Safety holds
under asynchrony. Liveness assumes partial synchrony~\cite{dls88}: after an unknown GST,
messages between honest parties arrive within $\Delta$. Configurations change through PDCC
(\S\ref{sec:reconfig}).

\paragraph{Accounts and shared resources.}
An account $a$ has a personal balance and a nonce. A \emph{shared resource} $\rho$ has a
balance, a version (the \emph{resource nonce}) and an authorization set $\Auth(\rho)$ of
accounts allowed to debit it, with $|\Auth(\rho)| \ge 2$ in practice.

\paragraph{Slots.}
A \emph{nonce slot} consists of a pair $(a, i)$ and a \emph{resource slot} consists of a pair $(\rho, v)$. A transaction consumes one sender nonce slot and, if it debits a shared resource, also one resource slot.

\paragraph{Transactions.}
A transaction $T$ is signed by its sender $a$ and carries a nonce $i$, so
$\nslot(T) = (a, i)$. It is of one of two kinds.
\begin{itemize}[noitemsep]
  \item A \emph{personal} transaction debits $a$'s personal balance.
  \item A \emph{shared} transaction debits a shared resource $\rho$ at version $v$, requires
        $a \in \Auth(\rho)$, and has $\rslot(T) = (\rho, v)$.
\end{itemize}
Both kinds may credit any accounts or resources. Credits are not versioned.

\paragraph{Execution order.}
A validator executes the transactions of account $a$ in nonce order and the transactions of
resource $\rho$ in version order. A nonce slot is \emph{consumed} when a transaction on it is
executed. A resource slot is \emph{resolved} when a transaction on it is executed or the slot
is released; resolving $(\rho, v)$ opens $(\rho, v+1)$.

\paragraph{Properties.}
\begin{description}[noitemsep]
  \item[Consistency.] No two executed transactions share a slot. All honest validators
        resolve every resource slot the same way.
  \item[No overdraft.] No balance in an honest validator's state becomes negative.
  \item[Validity.] An executed transaction is signed by its sender, the sender is authorized
        for the resource it debits, and an honest validator found the debit covered.
  \item[Finality.] A final transaction is eventually executed by every honest validator, and
        no transaction conflicting with it is ever executed.
  \item[Liveness.] After GST: a personal transaction from an honest account is eventually final;
        every resource slot on which an honest validator saw an echo is eventually
        resolved; and an honest account whose shared transaction loses can reuse its nonce.
\end{description}

\section{When Client-Driven Certification Is Complete}
\label{sec:completeness}

In client-driven certification, a transaction is final when a quorum of validators has signed
it. Validators send their signatures to whoever collects them and need not see each other's.

\begin{lemma}
\label{lem:fastpay}
If honest validators sign at most one transaction per slot and never release a slot, then
client-driven certification is safe even if certificates are known only to clients.
\end{lemma}
\begin{proof}
Two certificates for conflicting transactions would come from two quorums, which share an
honest validator that signed both on the same slot.
\end{proof}

Note that this guarantee depends on never releasing a slot. If a slot is at some point released and re-signed, it becomes possible to contradict a certificate created but never shown to validators.

\begin{proposition}
\label{prop:impossible}
Let $n=3f+1$, $f\ge1$, and let a transaction be final once a quorum has signed it, even if
signatures are additionally broadcast among validators. Then no function $\Select$ that maps the
reports of a quorum of validators on a slot (each stating which transaction the reporter
signed) to a transaction or $\nil$ can ever release or reassign the slot without
contradicting some final transaction.
\end{proposition}
\begin{proof}
Let $f=1$ with honest $h_1,h_2,h_3$ and Byzantine $b$, and conflicting $A,B$ on one slot. In
execution $X$, $h_1,h_2$ sign $A$ and $h_3$ signs $B$. $b$ broadcasts a signature for $B$ and
privately gives the client of $A$ a signature for $A$, so $\{h_1,h_2,b\}$ certifies $A$. In
execution $Y$, $h_1$ signs $A$, $h_2,h_3$ sign $B$, and $b$ signs $B$, so $\{h_2,h_3,b\}$ certifies $B$.
The quorum $\{h_1,h_3,b\}$ reports $A,B,B$ in both executions, and every broadcast it
received is the same in both. $\Select$ must output $A$ in $X$ and $B$ in $Y$, a
contradiction. For general $f$, let $h_1$, $h_3$ and $b$ each stand for a group of $f$
validators and $h_2$ for one validator.
\end{proof}

Personal transactions are never released, so Lemma~\ref{lem:fastpay} covers them. Shared
transactions must sometimes be released, so Proposition~\ref{prop:impossible} applies to them
and they need visible certification.

\section{Protocol}
\label{sec:protocol}

\paragraph{Validator state.}
For each nonce slot $\sigma$, validator $p$ stores $\mathsf{lock}_p[\sigma]$, the transaction
that currently holds the slot at $p$, or $\nil$. For each resource slot $s$ it stores
$\mathsf{echoed}_p[s]$, $\mathsf{pc}_p[s]$ (a stored echo quorum or $\nil$),
$\mathsf{frozen}_p[s]$, and the resolution. It also stores the set $\mathsf{dead}_p$ of transactions known
to have lost (\S\ref{sec:unlock}).

\paragraph{Admission.}
Validator $p$ \emph{admits} transaction $T$ if all of the following hold:
\begin{enumerate}[noitemsep]
  \item $T$ is signed by its sender $a$, and $a \in \Auth(\rho)$ if $T$ debits $\rho$;
  \item $T \notin \mathsf{dead}_p$;
  \item \emph{user nonce:} nonce $i-1$ of $a$ is consumed at $p$, and
        $\mathsf{lock}_p[(a,i)] \in \{\nil, T\}$;
  \item \emph{resource nonce} (shared only): $(\rho, v-1)$ is resolved at $p$, and both
        $\mathsf{echoed}_p[(\rho,v)] = \nil$ and $\neg\mathsf{frozen}_p[(\rho,v)]$ hold;
  \item the debited balance in $p$'s view covers the amount and fees.
\end{enumerate}
On admission, $p$ sets $\mathsf{lock}_p[(a,i)] \gets T$ and, for a shared transaction,
$\mathsf{echoed}_p[(\rho,v)] \gets T$, in one atomic step.

\subsection{Personal path}

The personal path is FastPay. The client sends $T$ to all validators; each validator that
admits $T$ returns $\langle\VOTE,T\rangle$; a quorum of votes is a certificate $\Cert(T)$,
which is final; validators execute $T$ on receiving $\Cert(T)$. Nonce locks taken by personal
transactions are never released.

\subsection{Shared path}

\begin{enumerate}[noitemsep]
  \item \textbf{Submit.} The client sends $T$ to any validator. If it does not see $T$ become
        final within a timeout, it sends $T$ to another validator.
  \item \textbf{Echo.} A validator that receives $T$ (from a client or inside an $\ECHO$)
        and admits it broadcasts $\langle\ECHO, T\rangle_p$ to all validators, including $T$.
  \item \textbf{Ready.} A validator $p$ that holds $\ECHO$s for $T$ from a quorum, with
        $\neg\mathsf{frozen}_p[\rslot(T)]$, stores the quorum as
        $\mathsf{pc}_p[\rslot(T)] \gets \PC(T)$ and broadcasts $\langle\READY, T\rangle_p$.
  \item \textbf{Final.} $\READY$s for $T$ from a quorum form $\RC(T)$, and $T$ is final.
        Validators also send $\READY$s to the client and recipients, who can verify $\RC(T)$.
        A validator executes $T$ on holding $\RC(T)$.
\end{enumerate}
A transaction on the shared path is therefore final after three message delays: submit, echo,
ready. It involves no leader and no consensus.

\begin{definition}[Final]
A personal transaction is final when $\Cert(T)$ exists. A shared transaction is final when
$\RC(T)$ exists.
\end{definition}

\section{Fast Unlock}
\label{sec:unlock}

\paragraph{Triggers.}
Unlocking is triggered inside the protocol. Validator $p$ freezes resource slot $s$ when any of
the following holds and $s$ is unresolved:
\begin{enumerate}[noitemsep]
  \item \emph{Timeout:} $\tau$ has elapsed since $p$ first saw an $\ECHO$ on $s$. A validator
        forwards the first $\ECHO$ it sees on each slot, so every honest validator starts the
        timer.
  \item \emph{Split:} $p$ holds $\ECHO$s on $s$ from a quorum of distinct validators and no
        transaction among them has a quorum; it may wait a short grace period first.
  \item \emph{Amplification:} $p$ received $\FREEZE$ reports on $s$ from $f+1$ validators.
\end{enumerate}
Triggers affect only liveness and latency. A Byzantine validator cannot freeze a slot at an
honest validator alone, because amplification needs $f+1$ reports.

\paragraph{Freeze.}
To freeze $s$, $p$ sets $\mathsf{frozen}_p[s] \gets \mathit{true}$ and broadcasts
\[
  \langle\FREEZE, s, \mathsf{echoed}_p[s], \mathsf{pc}_p[s]\rangle_p .
\]
Afterwards $p$ neither echoes on $s$ nor sends $\READY$ for any transaction on $s$. Freezing is
irrevocable.

\paragraph{Selection.}
A \emph{freeze set} $F$ for $s$ consists of well-formed $\FREEZE$ reports on $s$ from a quorum of
distinct validators.
\[
  \Select(F) =
  \begin{cases}
    T & \text{if some report in } F \text{ contains a valid } \PC(T),\\
    \nil & \text{otherwise.}
  \end{cases}
\]

\paragraph{Agreement.}
Validator $p$ proposes $(F_p, \Select(F_p))$ to the agreement instance $\mathsf{VBA}_s$ with
predicate $\Valid(F,d) \iff F$ is a freeze set for $s$ and $d = \Select(F)$. Any validated
agreement with a transferable decision proof $\pi$ works. In a PDCC system, a
dedicated instance of its leaderless PBFT-style settlement protocol serves, since any validator may
propose. Pending unlocks can also be batched into the next settlement round, at the cost of
waiting for it.

\paragraph{Resolution.}
On deciding $(F^*, d^*)$ with proof $\pi$, validators forward $\pi$ and proceed as follows.
\begin{itemize}[noitemsep]
  \item If $d^* = T$, they execute $T$. It consumes $T$'s nonce slot and resolves $s$.
  \item If $d^* = \nil$, they release $s$: the balance is unchanged and $(\rho, v+1)$ opens.
  \item Every transaction $T' \ne d^*$ on $s$ is \emph{dead}. Each validator adds it to
        $\mathsf{dead}$ and, if $\mathsf{lock}[\nslot(T')] = T'$, resets the nonce lock to
        $\nil$. The sender can reuse nonce $i$ for a new transaction.
\end{itemize}

\begin{algorithm}[t]
\caption{Shared path and fast unlock at validator $p$}
\label{alg:shared}
\begin{algorithmic}[1]
\Upon{receiving $T$ or $\langle\ECHO,T\rangle$ with $T$ admissible}
  \State take locks on $\nslot(T)$ and $\rslot(T)$; broadcast $\langle\ECHO,T\rangle_p$
\EndUpon
\Upon{$\ECHO$s for $T$ from a quorum and $\neg\mathsf{frozen}_p[\rslot(T)]$}
  \State $\mathsf{pc}_p[\rslot(T)] \gets \PC(T)$; broadcast $\langle\READY,T\rangle_p$
\EndUpon
\Upon{$\READY$s for $T$ from a quorum}
  \State execute $T$
\EndUpon
\Upon{a trigger for unresolved $s$}
  \State $\mathsf{frozen}_p[s] \gets \mathit{true}$; broadcast $\langle\FREEZE, s, \mathsf{echoed}_p[s], \mathsf{pc}_p[s]\rangle_p$
\EndUpon
\Upon{a freeze set $F_p$ for $s$}
  \State propose $(F_p, \Select(F_p))$ to $\mathsf{VBA}_s$
\EndUpon
\Upon{$\mathsf{VBA}_s$ decides $(F^*, d^*)$ with proof $\pi$}
  \State forward $\pi$; execute $d^*$ or release $s$; mark the others dead and free their nonce locks
\EndUpon
\end{algorithmic}
\end{algorithm}

\section{Correctness}
\label{sec:proofs}

We say a transaction has a \emph{quorum} if a quorum of validators admitted it, that is, if
$\Cert(T)$ or $\PC(T)$ exists.

\begin{lemma}[One echo per resource slot]
\label{lem:oneecho}
An honest validator echoes at most one transaction on each resource slot, and none after
freezing it.
\end{lemma}
\begin{proof}
$\mathsf{echoed}_p[s]$ is set on admission and never reset, and admission requires
$\neg\mathsf{frozen}_p[s]$, which is irrevocable.
\end{proof}

\begin{lemma}[Unique echo quorum]
\label{lem:uniquepc}
At most one transaction per resource slot has an echo quorum.
\end{lemma}
\begin{proof}
Two echo quorums share an honest validator, which contradicts Lemma~\ref{lem:oneecho}.
\end{proof}

\begin{lemma}[Visibility]
\label{lem:visible}
If a shared $T$ on $s$ is final, then $\Select(F) = T$ for every freeze set $F$ for $s$.
\end{lemma}
\begin{proof}
$\RC(T)$ and $F$ come from two quorums, which share an honest validator $q$. $q$ sends
$\READY$ only while $s$ is unfrozen and only after storing $\PC(T)$. Since freezing is irrevocable,
its report in $F$ contains $\PC(T)$, and $\Select(F)=T$ by Lemma~\ref{lem:uniquepc}.
\end{proof}

\begin{lemma}[No late finality]
\label{lem:nolate}
If a freeze set $F$ for $s$ exists and $\Select(F) \ne T$, then $T$ is never final.
\end{lemma}
\begin{proof}
Suppose $\RC(T)$ exists, formed at any time. Its quorum shares an honest validator $q$ with
$F$. Since $q$ is in $F$, it froze $s$. Since it sent $\READY$ for $T$, it did so before
freezing, so it had stored $\PC(T)$ and reported it in $F$. Then $\Select(F)=T$ by
Lemma~\ref{lem:uniquepc}.
\end{proof}

\begin{lemma}[Dead transactions]
\label{lem:dead}
A dead transaction is never final and never executed.
\end{lemma}
\begin{proof}
$T'$ is dead only if $\mathsf{VBA}_s$ decided $(F^*, d^*)$ with $d^* \ne T'$, where
$d^* = \Select(F^*)$. By Lemma~\ref{lem:nolate}, $T'$ is never final. The only other way to
execute a shared transaction is to be the decision, and the decision is unique by
agreement of $\mathsf{VBA}_s$.
\end{proof}

\begin{lemma}[Nonce slots]
\label{lem:nonce}
If $T \ne T'$ share a nonce slot and both have a quorum, at least one of them is dead.
\end{lemma}
\begin{proof}
The two quorums share an honest validator $q$, which admitted both. Admission sets
$\mathsf{lock}_q[\sigma]$ and requires it to be $\nil$ or the same transaction. So $q$ admitted
one of them, say $X$, then reset the lock before admitting the other. A lock is reset only
when the transaction holding it becomes dead, so $X$ is dead.
\end{proof}

\begin{theorem}[Resource slot agreement]
\label{thm:slot}
All honest validators resolve each resource slot $s$ the same way.
\end{theorem}
\begin{proof}
A validator resolves $s$ by executing a final $T$ or by applying $\mathsf{VBA}_s$'s decision.
Two final transactions on $s$ would have two echo quorums, contradicting
Lemma~\ref{lem:uniquepc}. The decision is unique, and it equals $\Select(F^*)$. If $T$ is
final, the decision is $T$ by Lemma~\ref{lem:visible}. If the decision is $d^*$, no $T \ne d^*$
becomes final later by Lemma~\ref{lem:nolate}.
\end{proof}

\begin{theorem}[Consistency]
\label{thm:consistency}
No two executed transactions share a slot.
\end{theorem}
\begin{proof}
For resource slots this follows from Theorem~\ref{thm:slot}. For a nonce slot, every
executed transaction has a quorum: a personal one through $\Cert(T)$, a final shared one
because $\RC(T)$ implies a stored $\PC(T)$, and a decided shared one because $\Select$ returns
a transaction only if a report contains its $\PC(T)$. If two executed transactions shared a nonce slot, one of them would be
dead by Lemma~\ref{lem:nonce}, and dead transactions are not executed by Lemma~\ref{lem:dead}.
\end{proof}

\begin{theorem}[Validity and no overdraft]
\label{thm:validity}
Every executed transaction was signed by its sender, the sender was authorized for the debited
resource, and the debit is executed only when the executing validator's balance covers it.
\end{theorem}
\begin{proof}
Every executed transaction has a quorum, so at least $f+1$ honest validators admitted it,
checking the signature, the authorization and the balance in their view. Debits of a balance
are ordered (nonce order for personal balances, version order for resources), and honest
validators agree on them by Theorems~\ref{thm:slot} and~\ref{thm:consistency}. Credits only
grow, and every credit an honest admitter counted was executed by it and is eventually
executed by everyone (Lemma~\ref{lem:propagation}). A validator executes a debit only when
its own view covers it, buffering it until then.
\end{proof}

\begin{lemma}[Propagation]
\label{lem:propagation}
If an honest validator executes $T$ or releases a slot, every honest validator eventually does
the same.
\end{lemma}
\begin{proof}
Each such step is backed by self-verifying evidence ($\Cert(T)$, $\RC(T)$ or a decision
proof $\pi$), which honest validators forward. The evidence for a slot is applied after the
previous nonce and version. The validator that applied it had applied the previous ones, so
induction on nonce and version concludes.
\end{proof}

\begin{theorem}[Finality]
\label{thm:finality}
A final transaction is eventually executed by every honest validator, and no transaction that
shares a slot with it is ever executed.
\end{theorem}
\begin{proof}
By Lemma~\ref{lem:propagation} and Theorem~\ref{thm:consistency}.
\end{proof}

\begin{theorem}[Liveness]
\label{thm:liveness}
Assume partial synchrony and that $\mathsf{VBA}$ terminates after GST.
\begin{enumerate}[noitemsep]
  \item A personal transaction from an honest account becomes final.
  \item A resource slot on which an honest validator saw an $\ECHO$ is eventually resolved.
  \item If an honest account's shared transaction $T$ is dead, every honest validator eventually
        releases the nonce lock that $T$ holds there.
\end{enumerate}
\end{theorem}
\begin{proof}
(1) An honest account never signs two transactions with one nonce, so after GST every honest
validator admits $T$ once its previous nonce is consumed and the balance it relies on has
arrived (Lemma~\ref{lem:propagation}). The $2f+1$ honest votes form $\Cert(T)$.

(2) The first honest validator to see an $\ECHO$ on $s$ forwards it, so all honest validators
start the timer. If $s$ is not resolved within $\tau$ after GST, every honest validator
freezes, each receives at least $2f+1$ $\FREEZE$ reports, and each proposes a valid value.
Termination of $\mathsf{VBA}_s$ resolves $s$.

(3) The decision proof $\pi$ reaches every honest validator, each of which marks $T$ dead and
resets $\mathsf{lock}[\nslot(T)]$ if it holds $T$.
\end{proof}

\begin{remark}[Accounts that equivocate]
An account that signs two personal transactions with the same nonce may lock its own nonce
forever, as in FastPay. If it signs a personal and a shared transaction with the same nonce,
the shared one may die and release its locks, while the personal one keeps its own.
Lemma~\ref{lem:nonce} still prevents both from executing. In neither case is another
account's resource affected beyond one released version.
\end{remark}

\section{Cost}
\label{sec:cost}

\begin{table}[h]
\centering
\begin{tabular}{@{}lccc@{}}
\toprule
Path & Messages per transaction & Delays to finality & Consensus \\
\midrule
Personal (FastPay) & $O(n)$ & 2 (certificate at client) & never \\
Shared, no contention & $O(n^2/B)$ with batches of $B$ & 3 & never \\
Shared, contended slot & as above + one $\mathsf{VBA}$ & $\tau$ + 1 + $\mathsf{VBA}$ & that slot only \\
\bottomrule
\end{tabular}
\caption{Cost of each path.}
\label{tab:cost}
\end{table}

The quadratic message pattern of the shared path is the price of visible certification
(Proposition~\ref{prop:impossible}), and only shared resources pay it. Two standard
techniques make it small in bytes.

\paragraph{Digests and batching.}
The validator that first receives a transaction broadcasts its body once. $\ECHO$ and
$\READY$ messages carry only digests and cover a batch of $B$ transactions. A validator that
receives an $\ECHO$ for an unknown digest fetches the body from the echoing validator.
Admission, freezing and selection are per slot. If a batch contains a transaction that a validator
does not admit, the validator echoes the batch with that position marked as rejected, so the
remaining transactions are unaffected.

\paragraph{A simple byte model.}
Let $|tx|$ be the size of the transaction, $h$ the size of the $\ECHO$ or $\READY$ header with its signature, and $\sigma$ the size of a signature. Then the approximate number of bytes sent per transaction by all of the validators is
\[
  \text{personal: } n\,|tx| + n\,\sigma + n\,|\Cert|,
  \qquad
  \text{shared: } n\,|tx| + \frac{2\,n^2 h}{B}.
\]
With aggregated signatures, certificates have constant size plus an $n$-bit signer bitmap. For
illustration, with $n=100$, $|tx| = 250$\,B, $h = 100$\,B and $B = 1000$, the shared-path
overhead is $2n^2h/B = 2$\,kB per transaction against $25$\,kB for disseminating the body, about
8\%. These figures are an analytical model, not measurements; measuring them on a
prototype is future work.

\paragraph{Computation.}
Per batch, each validator verifies $O(n)$ echo and ready signatures, so the per-transaction
verification cost is $O(n/B)$, plus one aggregate verification per received certificate.

\section{Reconfiguration with PDCC}
\label{sec:reconfig}

PDCC~\cite{pdcc} tags certificates with the configuration $c$ in which they are formed. A
validator accepts only certificates of its installed configuration (PDCC Invariant~I), and it
installs $c+1$ when the settlement consensus commits a checkpoint payload containing a
certified reconfiguration proposal. Applied naively, this rule can lose a final transaction:
a personal certificate or a ready certificate formed in $c$ may still be in transit, or held
only by a client, when $c+1$ is installed, and $c+1$ then rejects it. We close this gap with
a two-step transition. The first step seals $c$, and the second installs $c+1$ from a quorum of
closing reports. All objects below are tagged with $c$, and quorums are counted in $M_c$.

\paragraph{Seal.}
When the settlement consensus commits a checkpoint containing a valid reconfiguration
proposal for $c \to c+1$, every honest validator of $M_c$ \emph{seals} $c$. It stops
admitting transactions, sending $\READY$ and acknowledging anything in $c$, and freezes every
unresolved resource slot. Since all honest validators deliver the same committed checkpoint,
they all seal.

\paragraph{Closing report.}
After sealing, validator $p$ sends a signed closing report $\mathsf{CR}_p$ containing
\begin{itemize}[noitemsep]
  \item every transaction $p$ executed since the last committed checkpoint, with the
        certificate or decision proof that justified it;
  \item for every nonce slot, the non-dead transaction holding $\mathsf{lock}_p$, if any,
        together with $p$'s signature on it;
  \item for every unresolved resource slot $s$, its $\FREEZE$ content $(\mathsf{echoed}_p[s],
        \mathsf{pc}_p[s])$.
\end{itemize}

\paragraph{Install.}
The next settlement proposal that installs $c+1$ must contain closing reports $\mathcal{C}$
from a quorum of $M_c$; this is part of its validity predicate. The new state is computed
deterministically from the last checkpoint and $\mathcal{C}$.
\begin{enumerate}[noitemsep]
  \item Every transaction listed as executed in some report of $\mathcal{C}$, with valid
        evidence, is executed.
  \item For every unresolved resource slot $s$, the reports of $\mathcal{C}$ form a freeze set
        $F_s$. The slot is resolved as $\Select(F_s)$: that transaction is executed, or the slot
        is released. Transactions on $s$ other than $\Select(F_s)$ are dead.
  \item Every nonce slot $\sigma$ for which some report lists a lock becomes
        \emph{inherited} with the set $L_\sigma$ of non-dead transactions reported on it.
\end{enumerate}
One committed proposal resolves every frozen slot simultaneously. The settlement consensus acts as
the validated agreement for all of them, and no per-slot instance is needed at
reconfiguration.

\paragraph{Inherited nonce slots in $c+1$.}
For an inherited slot $\sigma$, validators of $M_{c+1}$ accept a personal certificate
$\Cert_c(T)$ from configuration $c$, verified against $M_c$ in the configuration history, if
and only if $T \in L_\sigma$. If $|L_\sigma| = 1$, they also admit the single transaction in
$L_\sigma$ in $c+1$, and no other transaction on $\sigma$. If $|L_\sigma| \ge 2$, the sender
equivocated, and the slot accepts only an old certificate for a member of $L_\sigma$, exactly
as an equivocated slot in FastPay. This relaxes PDCC Invariant~I for inherited slots only.
Every other certificate from $c$ is still rejected.

\begin{lemma}[No final transaction is lost]
\label{lem:carry}
If $T$ is final in $c$, then after installing $c+1$ either $T$ is executed or $T$ is a
personal transaction, $\nslot(T)$ is inherited and $T \in L_{\nslot(T)}$, so that
$\Cert_c(T)$ is accepted in $c+1$.
\end{lemma}
\begin{proof}
A final $T$ was admitted by a quorum, hence by at least $f+1$ honest validators. Every quorum
contains at least one of them, so $\mathcal{C}$ contains a report $\mathsf{CR}_q$ from an honest
$q$ that admitted $T$. All admissions in $c$ precede sealing, so $q$'s report reflects $T$:
either $q$ executed $T$ and lists it as executed, or $q$ still holds the lock.

If $T$ is shared and unexecuted at $q$, then $s=\rslot(T)$ is unresolved at $q$, and $F_s$ is
a freeze set. By Lemma~\ref{lem:visible}, $\Select(F_s)=T$, so $T$ is executed at
installation. If $T$ is personal and unexecuted at $q$, $q$'s lock on $\nslot(T)$ is $T$,
because personal locks are never reset, so $T \in L_{\nslot(T)}$.
\end{proof}

\begin{lemma}[No conflict across configurations]
\label{lem:crossconsistency}
No two executed transactions share a slot, across configurations.
\end{lemma}
\begin{proof}
Within $c$ this is Theorem~\ref{thm:consistency}. At installation, executed transactions carry
evidence that they are final or decided, and resource slots are resolved by $\Select$ on a
freeze set. Lemmas~\ref{lem:visible} and~\ref{lem:nolate} apply verbatim, because sealing is
a freeze of every unresolved slot. After installation, a nonce slot that is not inherited
had no lock at any reporter in $\mathcal{C}$. By the argument of Lemma~\ref{lem:carry}, no
transaction on it has a quorum in $c$. For an inherited slot, two transactions $T \ne T'$
executed on it would both belong to $L_\sigma$, or one would be admitted fresh in $c+1$ while
$|L_\sigma| = 1$. Two members of $L_\sigma$ with a certificate in $c$ would contradict
Lemma~\ref{lem:nonce}, since neither is dead. A fresh admission in $c+1$ is restricted to
the unique member of $L_\sigma$, which is the only one that can hold a certificate from $c$.
\end{proof}

\paragraph{Liveness and cost.}
Under PDCC's overlap assumption, at least $Q_c$ honest members of $M_c$ remain to send closing
reports, so the installing proposal can be formed. Between sealing and installation,
configuration $c$ admits no new transactions, so payments pause for one settlement round at
each reconfiguration. Shortening this pause, for instance by sealing per account range, is left
to future work.

\section{Mechanized Verification}
\label{sec:eval}

We model the protocol in TLA$^+$~\cite{lamport2002} and check it exhaustively with the TLC model checker. The specification, configurations and run script are open source at \url{https://github.com/sglk123/resource-lock-tla}.

\paragraph{Model.}
There are $n=4$ validators in this model, with three honest and one Byzantine ($f=1$). The Byzantine validator signs all messages (public and private) on all transactions and thus is included in all quorums. In a freeze set it may report any existing echo quorum. The workload consists of four transactions that cover all \S\ref{sec:model} conflict patterns:
\begin{itemize}[noitemsep]
  \item $A$ and $B$, shared transactions from different accounts racing on the same resource
        version;
  \item $P$, a personal transaction with the same sender nonce as $A$, so that the account
        equivocates across the two paths;
  \item $C$, a shared transaction on another resource with the same sender nonce as $A$, which
        reuses the nonce after $A$ loses.
\end{itemize}
Unlock triggers are abstracted away: any honest validator may freeze any slot at any
time. This over-approximates the triggers of \S\ref{sec:unlock}, which only restrict when a
freeze happens. Validated agreement is abstracted as a single atomic decision
$\Select(F)$ over some freeze set $F$, which captures its agreement and external-validity
properties.

\paragraph{Properties.}
We check the conjunction of: \emph{Consistency} (no two executed transactions share a slot);
\emph{finality safety} (no executed transaction conflicts with a final one); \emph{unlock
respects finality} (the decision on a slot is never a transaction other than a final one);
\emph{dead transactions are never final}; and \emph{at most one echo quorum per resource slot}.

\paragraph{Results.}
\begin{table}[h]
\centering
\begin{tabular}{@{}lrrl@{}}
\toprule
Variant & Distinct states & Depth & Result \\
\midrule
Protocol ($\ECHO$ + $\READY$) & 41{,}333{,}964 & 30 & no violation (exhaustive) \\
Finality at bare echo quorum & --- & 6 & violation found \\
\bottomrule
\end{tabular}
\caption{TLC results ($n=4$, $f=1$, four transactions).}
\label{tab:tlc}
\end{table}
The full protocol satisfies all properties over the complete reachable state space; TLC
generated $346$ million states, of which $41.3$ million are distinct, and terminated in about
six minutes on a laptop. In the variant in which a shared transaction is final once an echo
quorum exists, TLC finds a counterexample of depth six that matches
Proposition~\ref{prop:impossible}. Two honest validators echo $A$, which together with the
Byzantine validator's private signature is an echo quorum, so $A$ is final. One of the two
freezes before storing that quorum, and the unlock decides $\nil$, contradicting $A$.

\paragraph{Scope.}
Model checking covers this instance exhaustively, not all $n$ and workloads; the general
case is established by the proofs of \S\ref{sec:proofs}. Reconfiguration
(\S\ref{sec:reconfig}) and liveness are not part of the current model and are covered by
the proofs only.

\section{Related Work}
\label{sec:related}

FastPay~\cite{fastpay} introduced client-driven certificate payments; Astro~\cite{astro}
payments over Byzantine reliable broadcast. Guerraoui et al.~\cite{guerraoui2019} give the
consensus number of asset transfer. Sui~\cite{sui} separates owned objects (consistent
broadcast) from shared objects (consensus) and locks equivocated owned objects until the
end of the epoch. Bracha~\cite{bracha1987} introduced the echo/ready structure we use per
resource version. Fast Byzantine consensus~\cite{fab,kuznetsov2021} studies the
committee size needed for deciding from a single round of votes, of which
Proposition~\ref{prop:impossible} is a payment-specific instance. PDCC~\cite{pdcc} and dynamic
BFT~\cite{dynamicbft} handle reconfiguration.

\section{Conclusion}

Client-driven certification is complete exactly when locks are never undone. Shared balances
require undoing locks, and therefore require the protocol itself to record certification.
Our hybrid keeps personal payments on the FastPay path. It certifies shared-resource payments by
a leaderless echo/ready broadcast with dual-nonce admission, and resolves contention with
agreement on the single contended version. Agreement is paid for only where writers
actually contend.

\bibliographystyle{plain}
\bibliography{refs}

\end{document}